\documentclass[10pt,letterpaper,twocolumn,twoside,conference]{IEEEtran}

\usepackage{cite}
\usepackage{times}
\usepackage{dsfont}
\usepackage{cite}
\usepackage{times}

\usepackage{graphicx}
\usepackage{subfigure}
\usepackage{color}
\usepackage{ifpdf}
\usepackage{epsfig}
\usepackage{latexsym}
\usepackage{amsfonts}
\usepackage{amssymb}
\usepackage{paralist}
\usepackage{comment}
\usepackage{xspace}
\usepackage{mathrsfs}
\usepackage{amssymb}
\usepackage{slashbox}
\usepackage{color}
\usepackage[small]{caption2}
\usepackage[ruled]{algorithm}
\usepackage[noend]{algorithmic}

\usepackage{amssymb}
\usepackage{url,epsfig,array}
\usepackage{leftidx}
\usepackage{amsmath}

\def\ie{\textit{i.e.}\xspace}
\def\etal{\textit{et al.}\xspace}
\def\etc{\textit{etc.}\xspace}

\def\eg{\textit{e.g.}\xspace}

\newtheorem{theorem}{Theorem}

\renewcommand{\algorithmicrequire}{\textbf{Input:}}
\renewcommand{\algorithmicensure}{\textbf{Output:}}

\begin{document}
\bibliographystyle{plain}
%
\title{Near-Optimal Truthful Auction Mechanisms in Secondary Spectrum Markets}

\author{\IEEEauthorblockN{Yu-e Sun\IEEEauthorrefmark{1}, He Huang\IEEEauthorrefmark{2}, Xiang-yang Li\IEEEauthorrefmark{3}\IEEEauthorrefmark{4}, Zhili Chen\IEEEauthorrefmark{5}, Wei Yang\IEEEauthorrefmark{5}, Hongli Xu\IEEEauthorrefmark{5}, and Liusheng Huang\IEEEauthorrefmark{5}}
\IEEEauthorblockA{\IEEEauthorrefmark{1}School of Urban Rail Transportation, Soochow University, Suzhou, China\\
\IEEEauthorrefmark{2}School of Computer Science and Technology, Soochow University, Suzhou, China\\
\IEEEauthorrefmark{3}Department of Computer Science, Illinois Institute of Technology, Chicago, USA\\
\IEEEauthorrefmark{4}Tsinghua National Laboratory for Information Science and Technology (TNLIST), Tsinghua University\\
\IEEEauthorrefmark{5}Department of Computer Science and Technology, University of Science and Technology of China, Hefei, China\\
Emails: \{sunye12, huangh\}@suda.edu.cn, xli@cs.iit.edu, \{zlchen3,qubit,xuhongli,lshuang\}@ustc.edu.cn}}

\maketitle

\begin{abstract}
In this work, we study spectrum auction problem where each request from secondary users has spatial, temporal, and spectral features.
With the requests of secondary users and the reserve price of the primary user, our goal is to design truthful
mechanisms that will either maximize the social efficiency or maximize the revenue of the primary user.
As the optimal conflict-free spectrum allocation problem is NP-hard, in this work, we design near optimal spectrum
allocation mechanisms separately based on the following techniques: derandomized allocation from integer programming formulation,
its linear programming (LP) relaxation, and the dual of the LP. We theoretically prove that 1) our near optimal allocation methods are
bid monotone, which implys truthful auction mechanisms; and 2) our near optimal allocation methods can achieve a social
 efficiency or a revenue that is at least $1-\frac{1}{e}$ times of the optimal respectively.
At last, we conduct extensive simulations to study the performances
(social efficiency, revenue) of the proposed methods, and
the simulation results corroborate our theoretical
analysis.

\end{abstract}


\section{Introduction}
\label{sec:intro}
The growing demand for limited spectrum resource
poses a great challenge in spectrum allocation and usage \cite{chen2009mining}.
One of the most promising methods is spectrum auction, which provides
sufficient incentive for primary user (\emph{a.k.a seller}) to sublease spectrum to
secondary users (\emph{a.k.a buyers}).
The design of spectrum auction mechanisms are facing two major
challenges. First, spectrum channels can be reused in spatial, temporal, and spectral
domain if the buyers are conflict-free with each other, and thus, allocating the requests of buyers in channels optimally is often an NP-hard problem.
Second, truthfulness is regarded as one of the most critical properties, however, it's difficult to
ensure truthfulness in a spectrum auction mechanism with performance guarantee.

Many mechanisms were proposed to address some of these challenges.
For example, \cite{xu2010secondary,wangdistrict,fengtahes,zhou2009trust,zhou2008ebay,zhutruthful,jia2009revenue,gopinathanstrategyproof,wu2011small}
focused on truthfulness and spatial reuse. In \cite{zhou2008ebay}, truthful mechanism is
designed for the spectrum spatial reuse, but no performance guarantee on social efficiency and
 revenue. \cite{al2011truthful} and
 \cite{jia2009revenue} focused on maximizing revenue for the auctioneer, and the social efficiency maximization problem was studied in  \cite{gopinathanstrategyproof} and \cite{zhutruthful}. However, these results did not
  consider the temporal reuse of the spectrum. In practice, the buyer will employ temporal
reuse  to improve the utilization of spectrums. Following in this direction, temporal
reuse was considered in~\cite{xu2010salsa,dong2012combinatorial,wang2010toda,xu2011tofu,deek2011preempt}.
However, most of the studies that focused on temporal reuse assume that the conflict graph of buyers' geometry locations
is a completed graph for each channel, which means there is no
spatial reuse.
Feng \etal \cite{fengtahes} studied the case that spectrum can be reused both in spatial
and in temporal domains, and proposed a truthful double auction mechanism for spectrum. Nevertheless, performance
guarantee is  neglected in \cite{fengtahes}.
To the best of our knowledge, there is no truthful spectrum
auction mechanism with performance guarantee, in which spectrums can be reused both in spatial and temporal domains.

To tackle this challenge, we propose a truthful auction framework
for practical spectrum markets in this paper. Maximization of the \emph{\textbf{social efficiency}},
\ie allocating a channel to buyers who \textbf{value} it most, and
maximization of the \emph{\textbf{expected revenue}}, \ie
allocating a channel to buyers who \textbf{pay} it most, both are the
nature goals for spectrum auctions.
Thus, we aim at designing a framework that can flexibly choose the
 optimization goal between social efficiency
 and expected revenue.
Then, we propose a set of channel allocation mechanisms which
can either maximize the social efficiency or the expected revenue, and payment
schemes which ensure the truthfulness of our framework. To the best of
our knowledge, we are the
first to design truthful spectrum auctions while considering spatial and temporal reuse with
performance guarantee.

In this work, we design a
framework for spectrum auction
which can flexibly choose the optimization goal and the channel
allocation and payment mechanism.
Since channels can be reused in both spatial and temporal domain, the
problem of allocating requests of buyers
in channels optimally to maximize the social efficiency or the
expected revenue is NP-hard. It is
more complex than the problem of matching requests and channels
optimally only in spatial or
temporal domain. To tackle this challenge, we first relax the
integer programming formulation of the channel allocation
problem into a linear program (LP) problem,
which is solvable in polynomial time. A fractional solution for
channel allocation can be
obtained by solving this LP optimally. Then, we transform this fractional solution into a
feasible integer solution of the original channel allocation problem
by using a carefully designed randomized rounding procedure that
ensures the feasibility of the solution and good approximation to the
objective functions. We prove that the \textbf{expected} weight of the feasible integer
solution is at least $1-1/e $ times of the weight of the optimal solution. To complete our allocation
mechanism, we also propose a derandomization algorithm to get a feasible solution whose weight
is always guaranteed to be  at least  $1-1/e $ times of the weight of the optimal solution.
Then, we propose a revised derandomization algorithm and prove that this new
allocation method does satisfy the bid-monotone property, thus, implying a truthful mechanism.
To ensure that the payment mechanism is also solvable in polynomial time,
we further design a channel allocation and payment mechanism CATE, which is truthful in expectation. We prove that the expected weight of CATE's solution is
larger than $1-1/e$ times of the optimal.
We point out that our
allocation mechanisms can either approximate the social efficiency or
the expected revenue, but not both simultaneously.

The rest of paper is organized as follows:
Section \ref{sec:prelim} introduces the preliminaries and our design
 objectives.
Section \ref{sec:framework} presents our mechanism design framework
for optimizing the social efficiency or the revenue of the seller.
In Section \ref{sec:lp}, we
 propose our allocation algorithm, which is based on derandomization of
 solutions from linear programming.
Section \ref{sec:simulation} presents our extensive simulations for evaluating the social efficiency, revenue, and spectrum
 utilization efficiency of our methods. We discuss the related literatures in Section \ref{sec:review}, and conclude the paper in Section \ref{sec:conclusion}.

\section{Preliminaries}\label{sec:prelim}

\subsection{Spectrum Auction Model}

Auctions in our model are executed periodically. In each round, the primary user subleases the access right of $m$ channels in the fixed areas
 during time interval $[0,T]$, and $n$ buyers request the usage of channels in fixed time intervals and geographical locations/areas. Our goal is
  to allocate these requests of buyers in channels, such that we maximize either the social efficiency or the expected revenue.

Assume each channel provided by the primary user has a set of conflict-free license areas, and the primary user only sells
the rights to access his under-used channels in their license areas.
Moreover, these license areas between different channels may have
intersections. To make our model more general, we consider two models of the requests of buyers. The first one is the \emph{Point model},
in which each buyer requests the usage of channels in a particular geographical location and during a fixed time interval.
 The second one is \emph{Area model}, in which each buyer
 requests the usage of channels in a particular geographical area and also during a fixed time interval.

We use $\mathcal{S}$ to denote the set of channels, and define each channel $s_j \in
\mathcal{S}$ as $s_j =(R_j ,A_j )$, where $A_j$ is its license area, and  $R_j$ is the interference radius of a transmission when a user transmits in channel $s_j $.
Let $\mathcal{B}$ be the set of buyers, in which each buyer $i \in \mathcal{B}$ is assumed to have a
request $r_i$. Let $\mathcal{R}$ be the set of requests of buyers. Each request $r_i\in \mathcal{R}$ is defined as $r_i =(L_i ,
 b_i, v_i ,a_i ,t_i ,d_i)$, where $L_i $ is $i$'s geographical location in \emph{Point model}
  or the area where $i$ wants to access the channel in
 the \emph{Area model}, $b_i $ denotes its bidding
price, $v_i$ stands for its true valuation, and $a_i $, $d_i $, and $t_i $ denote the arrival time, deadline, and
duration time (or time length), respectively. In this paper, we only consider the case of
$d_i -a_i =t_i $, which means that the time request from the buyer is a fixed time
interval. We leave the case of $d_i -a_i >t_i $ as the future work.

We say that two requests $r_i $ and $r_k $ conflict with each other, if they
satisfy the following constrains: (1) the distance between $L_i $ and $L_k $
is smaller than twice of the interference radius in the Point model, or $L_i\bigcap L_k \ne \emptyset$ in the Area model; and (2) the required
time intervals from $r_i $ and $r_k $ overlap with each other. We denote the conflict relationships
among requests by a conflict graph $\mathcal{G}=(\mathcal{V}, \mathcal{E})$, where $\mathcal{V}$ is the set of requests of buyers, and edge
$(r_i,r_k) \in \mathcal{E}$ if requests $r_i$ and $r_k$ conflict with each other. Note that, for the same requests  $r_i$ and $r_k$, different interference radius of channels will lead to a different conflict relationship.
We use a matrix $Y=(y_{i,k,j} )_{n\times n\times m} $ to represent the conflict relationships in graph $\mathcal{G}$, in which if
requests $r_i $ and $r_k $ conflict with each other in channel $s_j$, $y_{i,k,j}=1$;
otherwise, $y_{i,k,j}=0$. Since the spectrum is a local resource, we also need
to define a location matrix $C=(c_{i,j} )_{n\times m} $ to represent whether
$L_i $ is in the license regions of channel $s_j $. $c_{i,j}=1$ if $L_i $ is in the license regions of channel $c_j
$; otherwise, $c_{i,j} =0$. Therefore, two requests $r_i $ and $r_k $
can share channel $s_j $ only if $y_{i,k,j}=0$, and $c_{i,j}=1$, $c_{k,j}=1$.

\subsection{Problem Formulation}
The objective of our work is to design a mechanism
satisfying truthfulness constraint, while maximizing the \emph{social efficiency} or \emph{revenue}.
An auction is said to be truthful if revealing the true
valuation is the \emph{dominant strategy} for each buyer, regardless of other
buyers' bids. It has been proved that
 an auction mechanism is truthful if its allocation algorithm is
monotonic and it always charges critical values from its buyers \cite{nisan2007algorithmic}.
The \emph{critical value} for a buyer is the minimum bid value, with which the buyer will win the auction. In our
problem definition, truthfulness implies two aspects:
\begin{compactenum}
\item Buyers report their true valuations for the spectrum channels (called \emph{\textbf{value-Truthfulness}})
\item Buyers report their true required time intervals (called \emph{\textbf{time-Truthfulness}}).
\end{compactenum}

\emph{\textbf{Social Efficiency Maximization:}} Social efficiency for an auction
mechanism $\mathcal{M}$ is defined as $\max \sum\nolimits_{r_i\in \mathcal{R}} {v_i x_i (\mathcal{M})}$,
where $x_i(\mathcal{M})=1$ if buyer $i$ wins; otherwise, $x_i(\mathcal{M})=0$.

\emph{\textbf{Revenue Maximization:}} The revenue of an auction is the total payment of buyers.
 An auction maximizing the revenue for the auctioneer is known
as an optimal auction in economic theory \cite{myerson1981optimal}. In the optimal auction, Myerson introduces the notion of virtual valuation $\phi _i (b_i )$ as
\begin{equation}
\label{eq1}
\phi _i (b_i )= b_i -\frac{1-F_i (b_i )}{f_i (b_i )}
\tag{1}
\end{equation}
where $F_i (b_i )$ is the probability distribution function of true
valuations of buyer $i$, and $f_i (b_i )= \frac{dF_i (b_i)}{db_i}$.

According to the theory of optimal auction\cite{myerson1981optimal}, maximizing revenue is equivalent to finding the optimal
solution of $\max\sum\nolimits_{r_i\in \mathcal{R}} {\phi _i (b_i )x_i (\mathcal{M})}$. Notice that $F_i $ should be regular
for each buyers $i$, that is, $\phi _i (b_i)$ is monotone
non-decreasing in $b_i $. In fact, this requirement is mild, and most natural
distributions of interest (uniform, exponential, Gaussian \etc) are regular.

\section{A Strategyproof Spectrum Auction Framework}\label{sec:framework}
In this section, we propose a general truthful spectrum auction
framework with the goal of maximizing social efficiency or revenue, as shown
in \emph{Algorithm \ref{alg:1}}. In our framework, we can flexibly choose different
optimization targets according to the practical requirements of auction
problems. The details are depicted as follows.

\renewcommand{\algorithmicrequire}{\textbf{Input:}}
\renewcommand{\algorithmicensure}{\textbf{Output:}}

{\small
\begin{algorithm}
\caption{Our truthful spectrum auction framework}\label{alg:1}
\begin{algorithmic}[1]

\REQUIRE ~conflict graph $\mathcal{G}$, location matrix $C$, set of channels $\mathcal{S}$, set of requests $\mathcal{R}$, monotone allocation and payment mechanism $\mathcal{A}$;\\

\ENSURE ~channel assignment $X$, payment $P$;

\STATE $\mathcal{R'}= \mathcal{R}$;
\FOR {each $r_i \in \mathcal{R}$}
\STATE $p_i=0$;
\IF {the target is maximization of social efficiency}
\STATE $\phi_i(b_i) = b_i$;
\ELSE
\STATE $\phi_i(b_i) =b_i - \frac{1-F_i (b_i)}{f_i (b_i)}$;
\IF {$\phi_i(b_i)< \eta^{\phi}t_i$}
\STATE $\mathcal{R'}= \mathcal{R'}/{r_i}$;
\ENDIF
\ENDIF
\ENDFOR
\STATE Run $\mathcal{A}$ using the set of virtual bids $\{\phi _i(b_i) \}_{r_i \in \mathcal{R'}}$;
\STATE Let $X= (x_i )_{r_i \in \mathcal{R'}}$ be the channel allocation and $\tilde{P} =(\tilde{{p}_i})_{r_i \in \mathcal{R'}}$ be the corresponding payment returned by $\mathcal{A}$;
\FOR {each $x_i = 1$}
\IF {the target is maximization of social efficiency}
\STATE $p_i = \tilde {p}_i$;
\ELSE
\STATE $p_i =\phi _i^{-1}(\tilde{p_i})$;
\ENDIF
\ENDFOR
\RETURN $(X,P)$;
\end{algorithmic}
\end{algorithm}
}

At the beginning of every auction period, we choose a particular
optimization target. If we choose the social efficiency maximization as our
target, we let the virtual bid $\phi _i(b_i) =b_i $. Then, we use the set of virtual
bids ${\bf{\Phi}}=(\phi _i(b_i) )_n $ as input to the channel allocation and payment
calculation mechanism $\mathcal{A}$. $\mathcal{A}$ returns an optimal or feasible
channel allocation ${\rm {\bf X}}{\kern 1pt}\ = (x_i )_n $, which
maximizes $\sum\nolimits_{r_i\in \mathcal{R}} {\phi _i(b_i) x_i } $. In ${\rm {\bf X}}$,
$x_i =1$ means that buyer $i$ wins the auction, while $x_i =0$ means
it loses. Meanwhile, $\mathcal{A}$ also returns a corresponding payment vector ${\rm
{\bf \tilde {P}}}\mbox{=(}\tilde {p}_i )_n $, and we charge each buyer $p_i
=\tilde {p}_i $.

If we choose to maximize the revenue of the auctioneer, we
convert the bid of each buyer into its corresponding virtual bid by setting $\phi _i(b_i) =b_i -\frac{1-F_i (b_i)}{f_i (b_i )}$.
Then, we can use the same allocation mechanism $\mathcal{A}$
in the case of social efficiency maximization to maximize $\sum\nolimits_{r_i\in \mathcal{R}} {\phi _i(b_i) x_i } $.
To ensure the worst case profit, we assume that the primary user already set a virtual
reservation price $\eta^{\phi}$, which is the minimum virtual price for spectrums \emph{per unit time}.
We take the requests whose virtual bid is larger than $\eta^{\phi}t_i$ as the input of $\mathcal{A}$,
and get an allocation vector ${\rm {\bf X}}$ and the corresponding payment
vector ${\rm {\bf \tilde {P}}}$. Different from the former target, the
payment vector ${\rm {\bf \tilde {P}}}$ we get in this case is virtual
payments of buyers. Therefore, we need to convert the virtual payments back
into the actual payments for buyers by the conversion of $p_i =\phi_i^{-1} (\tilde {p}_i )$.

As we have seen, if mechanism $\mathcal{A}$ is a monotonic allocation, and it always charges each winning buyer
its critical value, the proposed auction framework is truthful. In the following,
 we will show how to design a monotonic allocation method $\mathcal{A}$, which charges winners their critical
values.

\section{Allocation Mechanism with Approximation Ratio
  (1-1/\textnormal{e})}\label{sec:lp}

In this section, we propose the channel allocation mechanisms $\mathcal{A}$ under
our spectrum auction framework. We first present an optimal solution to the channel allocation
problem that maximizes the total bids or virtual bids of secondary users.
However, solving this channel allocation problem optimally is NP-hard. To
address this, we further design a set of $(1-1/e)$ approximation mechanisms
which can be solved in polynomial time. By using the \emph{LP relaxation}
technique, we first propose a deterministic mechanism (DCA) to get a solution
whose weight is larger than $1-1/e$ times of the optimal.
Then, we design an revise version of DCA, which called MDCA, to make sure
that our channel allocation mechanism is bid monotone. To ensure that our payment mechanism
can also be solved in polynomial time, we
further design mechanism CATE, which is truthful in expectation.

\subsection{The Optimal Channel Allocation}

In the channel allocation problem, we need to match the requests and channels
optimally under their constraints. For each request $r_i$, it can only be
allocated in the time slice between $a_i$ and $d_i$. And for each channel
$s_j$, it can only allocate time slices to the requests which are in its entire
license area. Moreover, we can only allocate a fixed time slice to the
requests conflict-free with each other. In order to
simplify the matching model between requests and channels, we would like to
segment the available time of each channel into many time slices. Recall
that the available time of each channel is [0,T] in each auction
period. Then, we use the arrival time $a_i$ and deadline $d_i$ of each request $r_i$
to partition the time interval [0,T]. Each arrival time/deadline of
requests divides the time axis of one channel into two parts.
As shown in Figure \ref{segmentation}, the arrival time and deadline of requests $r_1$, $r_2$ and $r_3$
divide the time interval [0,T] into 7 time slices.
Suppose there are $n$ requests, we can easily get that the time interval [0,T] can be divided into no
more than $2n+1$ time slices.

\input{segmentation.TpX}

After the introduction of segmentation process, we give the detailed description of the channel allocation
problem. First, for each partitioned time slice derived from channel $s_j$, it can
only be allocated to the requests within the license area of channel $s_j$.
Let $x_{j,i}^l$ represent whether the \emph{$l$-th} time slice of channel $s_j$
is allocated to the request $r_i$, then we get a constraint $x_{j,i}^l \le
c_{i,j}$. Second, each time slice can only be allocated to requests
conflict-free with each other. Thus, we get another constraint
$\sum\nolimits_{k\ne i} {x_{j,k}^l } y_{i,k,j} + x_{j,i}^l \le
1$. Let $t_j^l $ be the length of \emph{$l$-th} time slice in channel $s_j $.
Modify $a_i $ to be the first time slice $r_i $ wants to use, and $d_i $ to be the
last time slice $r_i $ wants to use. Moreover, if we allocate request $r_i $
in channel $s_j $, the time assigned to request $r_i $ from channel
$s_j $ should be equal to the required time of request $r_i $, so we get
$\sum_{l = a_i }^{d_i } {x_{j,i}^l t_j^l } = t_i
x_{i,j}$. From the analysis above, the allocation problem can be
formulated as follows.
\begin{equation*}
\max \sum\nolimits_{s_j\in \mathcal{S}} {\sum\nolimits_{r_i\in \mathcal{R'}} {\phi _i(b_i)}x_{i,j}},
\tag{IP (1)}
\end{equation*}
subject to
\begin{equation*}
\begin{cases}
\sum\nolimits_{s_j\in \mathcal{S}} {x_{i,j} } \le 1, \forall r_i\in \mathcal{R'}  \\
x_{j,i}^l \le c_{i,j}, \forall s_j\in \mathcal{S}, \forall r_i\in \mathcal{R'}, \forall l \\
\sum\nolimits_{k\ne i} {x_{j,k}^l} y_{i,k} + x_{j,i}^l \le 1, \forall s_j\in \mathcal{S}, \forall r_i\in \mathcal{R'}, \forall l \\
\sum\limits_{l=a_i }^{d_i } {x_{j,i}^l t_j^l } =t_i x_{i,j}, \forall s_j\in \mathcal{S}, \forall r_i\in \mathcal{R'}\\
x_{i,j}\in \{0,1\}, \forall s_j\in \mathcal{S},\forall r_i\in \mathcal{R'} \\
x_{j,i}^l \in \{0,1\}, \forall s_j\in \mathcal{S},\forall r_i \in \mathcal{R'}, \forall l \\
\end{cases}
\end{equation*}
where $x_{i,j}$ stands for whether channel $s_j$ is allocated to
request $r_i$ or not, $y_{i,k,j}$ represents whether request $r_i$
conflicts with request $r_k$ or not.

If we can get the optimal solution of integer programming IP(1), we can apply the best-known VCG
mechanism to design a truthful auction mechanism. In VCG mechanism, the
winner determination is to maximize the sum of winning bids, and the payment from
each buyer is the opportunity cost that its presence introduces to all the
other players. Assume that ${\rm {\bf X}}_{opt} =(x_k )_n$ is the
optimal solution of IP(1), where $x_k =\sum\nolimits_{s_j\in \mathcal{S}}
{x_{k,j} }$, and ${\rm {\bf X}}_{opt}$ is the allocation vector. For each $x_i =1$, the corresponding payment $\tilde
{p}_i$ is:
\begin{equation}
\label{eq2}
\tilde{p}_{i} =\mathop{\max }\limits_{X_{-i}}
\sum\limits_{k\ne i}{x_k
\phi_{k}(b_k)} - \mathop{\max}\limits_{X} \sum\limits_{k\ne
i} {x_k \phi_{k}(b_k)}
\tag{2}
\end{equation}

VCG mechanism guarantees the monotone allocation, and it always
charges each winner its critical value, so the resulted auction mechanism is
truthful.
As the maximum weighted independent set problem is a special case of the channel allocation, thus we have

\begin{theorem}
\label{theo:NP}
The social efficiency maximization or revenue maximization channel allocation problem  is  NP-hard.
\end{theorem}
\begin{proof}
For a simple case that there is one channel in our model, we need to allocate requests in this channel with conflict graph $\mathcal{G}$ , and the goal of our channel allocation mechanism is to maximize the sum of the virtual bids. Then, the channel allocation problem is a maximum weighted independent set problem in this case, which is an NP-hard problem. \end{proof}

Then solving the integer programming IP(1) is an NP-hard problem, which implies that VCG mechanism
is difficult to be applied to actual auctions. In order to tackle the NP-hardness, we employ \emph{LP relaxation} methods,
and further design a set of polynomial time solvable channel allocation mechanisms with an approximation factor of $1-1/e$.

\subsection{(1-1/e)-Approximation  methods}

\emph{LP relaxation} technique can be introduced to solve NP-hard problems,
and it often leads to a good approximation algorithm. We release IP(1) to
linear programming LP(2) by replace $x_{i,j}\in \{0,1\}$ with $0\le x_{i ,j} \le 1 $, and
replace $x_{j,i}^l \in \{0,1\}$ with $0 \le x_{j,i}^l \le 1$.
The allocation problem can be reformulated as:

{\setlength{\abovedisplayskip}{1pt}
\setlength{\belowdisplayskip}{1pt}
\begin{equation*}
\setlength{\abovedisplayskip}{1pt}
\setlength{\belowdisplayskip}{1pt}
\max \sum\nolimits_{s_j\in \mathcal{S}} {\sum\nolimits_{r_i\in \mathcal{R'}} {\phi _i(b_i)}x_{i,j}}
\tag{LP(2)}
\end{equation*}
subject to
\begin{equation*}
\begin{cases}
\sum\nolimits_{s_j\in \mathcal{S}} {x_{i,j} } \le 1, \forall r_i\in \mathcal{R'} \\
x_{j,i}^l \le c_{i,j}, \forall s_j\in \mathcal{S}, \forall r_i\in \mathcal{R'} \\
\sum\nolimits_{k\ne i} {x_{j,k}^l } y_{i,k}+x_{j,i}^l \le 1, \forall s_j\in \mathcal{S}, \forall r_i\in \mathcal{R'}, \forall l \\
\sum\limits_{l=a_i }^{d_i } {x_{j,i}^l t_j^l }=t_i x_{i,j}, \forall s_j\in \mathcal{S}, \forall r_i\in \mathcal{R'} \\
0\le x_{i,j} \le 1, \forall s_j\in \mathcal{S}, \forall r_i\in \mathcal{R'} \\
0\le x_{j,i}^l \le 1, \forall s_j\in \mathcal{S}, \forall r_i\in \mathcal{R'}, \forall l \\
\end{cases}
\end{equation*}}

Recall that the number of time slices is no more than $2n+1$ for each channel, so LP(2) has a polynomial number of variables and
constraints, and can be solved optimally in polynomial time.

\subsubsection{Randomized Rounding}

Suppose $O_{LP2}$ is the optimal solution of LP (2), we apply a standard
randomized rounding on it to obtain an integral feasible solution
$f_{IP1} $ to IP (1). The rounding procedure is presented as follows.
\begin{compactenum}
\item Randomly choose a channel $s_j$, randomly choose a request $r_i$ with $x_{i,j} >0$, and set $x_{i,j} =1$;
\item If $x_{i,j} =1$, set $x_{k,j} =0$ for all requests $r_k $ with $y_{i,k,j}=1$;
\item If $x_{i,j} =1$, set $x_{i,k} =0$ for all channels with $k\ne j$.
\item Repeat steps 1 to 3 until all requests have been processed.
\end{compactenum}

Through the randomized rounding procedure above, the optimal solution of LP(2) is
converted into a feasible solution of IP(1). Let $w_{O_{LP2} } $ be the
weight of $O_{LP2} $, and let $E(w_{f_{IP1} } )$ be the expected weight of $f_{IP1}
$. We show in Theorem \ref{theo:lp2} that $E(w_{f_{IP1}})\ge (1-1/e)w_{O_{LP2}} $.

\begin{theorem}
\label{theo:lp2}
The expected weight of the rounded solution is at least $1-1/e$ times of the weight of the optimal solution to LP (2).
\end{theorem}
\begin{proof}
For each request $r_i$, let $H=\{s_j \in \mathcal{S}:x_{i,j} >0\}$ be the set of channels $s_j \in \mathcal{S}$ with $x_{i,j} >0$, and let $h=|H|$.
Clearly, $0\le h\le m$. The probability that request $r_i$ is not allocated in any channel by $f_{IP1}$ is $\prod\limits_{j=1}^h{(1-x_{i,j} )}$.
Let $q_i$ denote the probability that request $r_i$ is allocated in one of the $h$ channels by $f_{IP1}$. Then, we get that
 $q_i=1-\prod\limits_{j=1}^h {(1-x_{i,j} )}$. It's obvious that $E(w_{f_{IP1}} )=w_{O_{LP2} }$ when $h=0$ or 1.
 Thus, we only consider the case $h\ge 2$ in the following. In this case, $q_i $ is minimized when $x_{i,j} =x_{i}/h$. Then, we have $q_i \ge 1-(1-x_{i}/h)^{h}$, and

\begin{equation}
\label{eq3}
\frac{q_i}{x_i} \ge \frac{1}{x_i}(1-(1-(x_{i}/h)^h) \ge  (1-(1- 1/h)^h) \ge 1-\frac{1}{e}
\tag{3}
\end{equation}

The right side of the inequality is a monotonically decreasing function depending on $x_i $, with $0\le x_i \le 1$. Thus, it is minimized when $x_i=1$, and we have

\begin{equation}
\begin{array}{r}
\displaystyle \frac{q_i}{x_i} \ge \frac{1}{x_i}(1-(1-{x_i}/h)^h) \ge (1-(1- 1/h)^h) \\
\displaystyle     \ge 1-\frac{1}{e}+\frac{1}{32h^2} \ge  1-1/e
\tag{4}
\end{array}
\end{equation}

For each request $r_i$ with $q_i >0$, its
contribution in the expected weight of the rounded solution is $q_i \phi_i(b_i) $,
and that in the weight of the optimal solution of LP(2) is $x_i \phi _i(b_i) $. Then we have
$\frac{q_i \phi _i(b_i) }{x_i \phi _i(b_i) }\ge 1-\frac{1}{e}$.
Since this inequality holds for any request $r_i \in \mathcal{R'}$, and $E(w_{f_{IP1}}
) = \sum\nolimits_{r_i\in \mathcal{R'}} {q_i \phi _i(b_i) } $, $w_{O_{LP2}}
=\sum\nolimits_{r_i\in \mathcal{R'}} {x_i \phi _i(b_i)}$, we have
$E(w_{f_{IP1} } )\ge (1-1 \mathord{\left/ {\vphantom {1 e}} \right.
\kern-\nulldelimiterspace} e)w_{O_{LP2} }$.
 \end{proof}

We have shown that the expected weight of feasible solution $f_{IP1} $ of IP
(1) obtained by our randomized rounding is larger than $1-1/e$ times of the weight of
the optimal solution of LP (2). Obviously, the weight of the optimal
solution of LP (2), which is denoted by $w_{O_{LP2} } $, is larger than the optimal solution of IP
(1), which is denoted by $w_{O_{IP1} } $. Therefore, we can get that
\begin{theorem}
The expected weight of the rounded solution is at least $1-1/e$ times of the weight of the optimal solution to IP (1).
\end{theorem}

\subsubsection{Deterministic Methods}

The rounding procedure only makes sure that the expected weight of $f_{IP1} $ is larger
than $1-1/e$ times of the weight of $O_{LP2} $. What we need
is to find a feasible solution of IP (1) whose weight is exactly larger
than $1-1/e$ times
of the $w_{O_{LP2} }$. In the following, we show that the rounding procedure can be derandomized
and how the method of conditional probabilities can be used in our setting.

\begin{algorithm}
\caption{DCA: Derandomized Channel Allocation Based on Linear Programming}\label{alg:2}
\begin{algorithmic}[1]

\REQUIRE ~Conflict graph $\mathcal{G}$, location matrix $C$, set of
channels $\mathcal{S}$, set $\mathcal{R'}$ sorted in increasing order according to $a _i$;

\ENSURE ~channel assignment $X^\ast$ ;

\STATE Solve LP(2) optimally;
\STATE $E(w_{f_{IP1} } )=\sum\limits_{s_j\in \mathcal{S}}{\phi_i(b_i)(1-\prod\nolimits_{s_j\in \mathcal{S}}{(1-x_{i,j})})}$;
\FOR {$i = 1$ to $n$}
\IF {$x_i>0$}
\FOR {$j = 1$ to $m$}
\IF {$E(w_{f_{IP1} } )\leq E(w_{f_{IP1}}|i,j)$}
\STATE set $x_{i,j}=1$, $x_i=1$;
\STATE set all $x_{i,k}=0$ and $x_{i,k}^l=0$ if $k\neq j$;
\STATE set all $x_{k,j}=0$ and $x_{k,j}^l=0$ if $k\neq i$ and $y_{i,k,j}=1$;
\STATE \text{Break}
\ENDIF
\ENDFOR
\ENDIF
\IF {$x_i\neq 1$}
\STATE $x_i=0$;
\ENDIF
\ENDFOR \RETURN $X^\ast$;
\end{algorithmic}
\end{algorithm}

Let $E(w_{f_{IP1} } \vert r_i \rightarrow s_j)$ be the expected weight when request $r_i $
is allocated in channel $s_j $, and let $E(w_{f_{IP1} } \vert \tilde {i})$ be the
expected weight when request $r_i$ will not be allocated in any channel.
Next, we will show how our derandomize algorithm works. We first sort all the
requests by their arrival time $a_i$ in the ascending order. Let $x_i =\sum\nolimits_{j\in \mathcal{S}} {x_{i,j}}$,
and then scan all the requests
one by one to decide which request can be allocated in channels. When request
$r_i $ is considered, we scan all of the channels that are available
 for $r_i $ to check
if $r_i $ can be allocated in one of them. If $E(w_{f_{IP1} } \vert
r_i \rightarrow s_j)<E(w_{f_{IP1} } )$, set $x_{i,j} =0$; otherwise, allocate $r_i $ in
channel $s_j $, and set $x_{i,j} =1$, $x_i =1$, $x_{i,k} =0$ if
$k\ne j$. Meanwhile, if $r_i $ is allocated in channel $s_j $, we set $x_{k,j}^l =0$ if $y_{i,k,j} =1$.

Suppose $r_i $ is the first request that satisfies $x_i >0$ in the ordered requests.
Let $q_{i,j} $ denote the probability that request $r_i $ is allocated in
channel $s_j $ and let $q_{\tilde{i}}$ denote the probability that $r_i$ is not allocated in any channel. By the formula for conditional probabilities, we have

\begin{equation}
\label{eq5}
E(w_{f_{IP1}})= \sum\limits_{r_j\in S}{E(w_{f_{IP1}}|r_i \rightarrow s_j)q_{i,j}}+E(w_{f_{IP1}}|\tilde{i})q_{\tilde{i}}
\tag{5}
\end{equation}

In particular, there exists at least one conditional expectation in
$E(w_{f_{IP1} } \vert r_i \rightarrow s_1),\cdots ,E(w_{f_{IP1} } \vert r_i \rightarrow s_m),E(w_{f_{IP1} }\\
\vert \tilde {i})$, which is larger than $E(w_{f_{IP1}})$. If it is
$E(w_{f_{IP1} } \vert r_i \rightarrow s_j)\ge E(w_{f_{IP1} } )$, we allocate request $r_i $ in
channel $s_j $; otherwise, $E(w_{f_{IP1} } \vert \tilde {i})\ge E(w_{f_{IP1}
} )$ holds, reject request $r_i $, and set $x_{i,j} =0$ for each
$s_j \in \mathcal{S}$. This can be done since $E(w_{f_{IP1} } )=\sum\nolimits_{r_i\in \mathcal{R'}}
{\phi _i(b_i) q_i } $, and $q_i $ can be computed precisely by

\begin{equation}
\label{eq6}
q_i=1-\prod _{s_j\in \mathcal{S}} (1-x_{i,j})
\tag{6}
\end{equation}

Let $q_{r_i \rightarrow s_j,k}$ stand for the probability that request $r_k$ is allocated in a channel when request $r_i$ is allocated in $s_j$. Then $q_{r_i \rightarrow s_j,k}$ can be calculated by

\begin{equation}
\label{eq7}
q_{r_i \rightarrow s_j,k} =\left\{ {\begin{array}{l}
1-\prod\nolimits_{o\ne j} {(1-x_{k,o} )} ,y_{i,k,j} =1 \\
q_k , otherwise \\
\end{array}} \right.
\tag{7}
\end{equation}

For each request $r_i $, we can compute $E(w_{f_{IP1} } \vert r_i \rightarrow s_j)$
precisely as the follows

\begin{equation}
\label{eq8}
E(w_{f_{IP1}}|r_i \rightarrow s_j)=\phi_i(b_i) +\sum\nolimits_{k\ne i} {\phi _k (b_k)q_{r_i \rightarrow s_j,k}}
\tag{8}
\end{equation}

Given the selections in the prior requests, we can continue
deterministically to allocate other requests and do the same thing while
maintaining the invariant that the conditional expectation $E(w_{f_{IP1} }
)$, never deceases. After allocating all of the requests, we can get a feasible
solution of IP (1) whose weight is as good as $E(w_{f_{IP1} } )$, \ie
at least $(1-1/e)w_{O_{LP2} }$.

Since LP(2) can be solved in polynomial time, and we can allocate requests in channels with time complexity $O(nm)$
 by using the optimal solution of LP(2), we get that
\begin{theorem}
DCA can be executed in polynomial time.
\end{theorem}
\begin{proof}
As mentioned above, LP(2) can be solved in polynomial time. Then, we allocate requests in channels with time complexity $O(nm)$ in DCA
with the optimal solution of LP(2). This finishes the proof. \end{proof}

Recall that to ensure the truthfulness of our auction mechanism, the
allocation algorithm must be bid-monotone. This means that if request $r_i $ wins
the auction with bid $v_i$, it always wins with bid $b_i >v_i $. In
\emph{Algorithm \ref{alg:2}}, request $r_i $ wins in the auction only if there exists a channel
$s_j $ which satisfies $E(w_{f_{IP1} } \vert r_i \rightarrow s_j)\ge E(w_{f_{IP1} } )$.
However, it is hard to judge that if $E(w_{f_{IP1} } \vert r_i \rightarrow s_j)$ is still
larger than $E(w_{f_{IP1} } )$ when request $r_i $ increases its bid.
We cannot prove or disprove the bid-monotone property of the allocation method DCA. Thus, it is unknown whether we can design a truthful mechanism based on this method. In the rest of the section, we revise this method and show that the revised method does satisfy the bid-monotone property.

Since that there exists at least one of the conditional expectations between
$max_{s_j \in \mathcal{S}} E(w_{f_{IP1} } \vert r_i \rightarrow s_j)$ and $E(w_{f_{IP1} } \\
\vert \tilde {i})$, which is larger than $E(w_{f_{IP1} })$. Thus, if we allocate $r_i $ in the channel with the maximal
conditional expectation as long as $max_{s_j \in \mathcal{S}} E(w_{f_{IP1} } \vert r_i \rightarrow s_j)\ge E(w_{f_{IP1} } \vert \tilde {i})$, and do not
allocate $r_i $ in any channel otherwise, we can also get a feasible
solution of IP (1), whose weight is as good as $E(w_{f_{IP1} } )$.

This can be done since we can compute $E(w_{f_{IP1} } \vert i,j)$ and
$E(w_{f_{IP1} } \vert \tilde {i})$ precisely as follows:

\begin{equation}
\label{eq9}
E(w_{f_{IP1} } \vert r_i \rightarrow s_j)=\phi _i(b_i) +E_{k\ne i} (w_{{f}'_{IP1} } \vert r_i \rightarrow s_j)
\tag{9}
\end{equation}

where $E_{k\ne i} (w_{{f}'_{IP1} } \vert r_i \rightarrow s_j)$ is the expected weight of all
other requests when request $r_i $ has been allocated in channel $s_j $. We can
get it by allocating $r_i $ in channel $s_j $ first, and then solve LP
(2) optimally with other requests.

\begin{equation}
\label{eq10}
E(w_{f_{IP1} } \vert \tilde {i})=E_{\mathcal{R'}/r_i} (w_{{f}'_{IP1} } )
\tag{10}
\end{equation}

where $E_{\mathcal{R}/r_i} (w_{{f}'_{IP1} } )$ is the expected weight of all other
requests when request $r_i$ does not be allocated in any channel. We can get
it by solving LP (2) optimally with requests except $r_i $.

Based on the observation above, we give an revised version of
 Algorithm DCA as follows.

\begin{algorithm}
\caption{MDCA: Monotone Derandomized Channel Allocation Based on Linear Programming}\label{alg:4}
\begin{algorithmic}[1]

\REQUIRE ~Conflict graph $\mathcal{G}$, location matrix $C$, set of channels $\mathcal{S}$, set of $\mathcal{R'}$ sorted in increasing order according to $a _i$;

\ENSURE ~channel assignment $X^\ast$ ;

\STATE Solve LP(2) optimally;
\FOR {$i = 1$ to $n$}
\FOR {$j = 1$ to $m$}
\IF {$x_{i,j}>0$}
\STATE $E(w_{f_{IP1} } \vert r_i \rightarrow s_k)=max_{s_j \in \mathcal{S}} E(w_{f_{IP1} } \vert r_i \rightarrow s_j)$
\IF {$E(w_{f_{IP1} } \vert r_i \rightarrow s_k)\ge E(w_{f_{IP1} } \vert \tilde {i})$}
\STATE set $x_{i,j}=1$, $x_i=1$;
\STATE set all $x_{i,k}=0$ and $x_{i,k}^l=0$ if $k\neq j$;
\STATE set all $x_{k,j}=0$ and $x_{k,j}^l=0$ if $k\neq i$ and $y_{i,k,j}=1$;
\STATE \text{Break}
\ENDIF
\ENDIF
\ENDFOR
\IF {$x_i\neq 1$}
\STATE $x_i=0$;
\ENDIF
\ENDFOR \RETURN $X^\ast$;
\end{algorithmic}
\end{algorithm}

In MDCA, we first sort all of the requests by their arrival times in the ascending order, and then we scan all requests one by one to decide which
request can be allocated in channels. When request $r_i $ is considered, we compute $E(w_{f_{IP1} } \vert r_i \rightarrow s_j)$
for all channels $s_j \in \mathcal{S}$ that no request conflicting with it has been allocated in. We allocate $r_i $ in channel $s_k $ when $E(w_{f_{IP1} }
\vert r_i \rightarrow s_k)=max_{s_j \in \mathcal{S}}  E(w_{f_{IP1} } \vert r_i \\
\rightarrow s_j)\ge E(w_{f_{IP1} } \vert \tilde {i})$, and reject it otherwise. After the last request was considered in MDCA, we get a feasible solution of IP (1), whose weight is as good as $E(w_{f_{IP1} } )$.

\begin{theorem}
MDCA (see Algorithm 3) is bid monotone.
\end{theorem}
\begin{proof}
Suppose request $r_i $ wins the auction with the bid $b_i $, and it is allocated with the channel $s_j $, but it cannot be allocated in any channel with the bid $b_i>v_i$.
There are two possible cases.

Case 1: $max_{s_j \in S} E(w_{f_{IP1} } \vert r_i \rightarrow s_j)< E(w_{f_{IP1} } \vert \tilde {i})$ when $r_i$ bids some value $b_i$ with $b_i > v_i$. However, when $r_i $ increases its bid, clearly, $E_{k\ne i} (w_{{f}'_{IP1} } \vert r_k \rightarrow s_j)$ and $E(w_{f_{IP1} } \vert \tilde {i})$ keep
invariant, and $E(w_{f_{IP1} } \vert r_i \rightarrow s_j)\ge E(w_{f_{IP1} } \vert \tilde {i})$ always holds in this case. Thus, our hypothesis does not hold in this case.

Case 2: When $r_i$ is considered with the bid $b_i > v_i$, the channel $s_j$ has been occupied by request $r_l$, which conflicts with $r_i$. Obviously, $r_l$ is not allocated in $s_j$ when $r_i$ bids $v_i$. That means
$E(w_{f_{IP1} } \vert r_l \rightarrow s_j)<E(w_{f_{IP1} } \vert \tilde {l})$ or the channel $s_j$ has been occupied by other requests
which conflict with $r_l$ but conflict-free with $r_i$ when $r_l$ was considered. In the first subcase, the contribution of $r_i$ in $E(w_{f_{IP1} } \vert \tilde {l})$ is larger than the contribution in $E(w_{f_{IP1} } \vert r_l \rightarrow s_j)$. Then, the increment of $E(w_{f_{IP1} } \vert \tilde {l})$ is lager than that of $E(w_{f_{IP1} } \vert r_l \rightarrow s_j)$ when $r_i$ increases its bid from $v_i$ to $b_i$. Thus, $r_l$ cannot be allocated in $s_j$ when
 $r_i$ bids $b_i>v_i$.
 Assume that $r_k$ which conflicts with $r_l$ is allocated in $s_j$ when $r_i$ bids $v_i$ in the second subcase. However,
 the contribution of $r_i$ in $E(w_{f_{IP1} } \vert r_k \rightarrow s_j)$ is no less than the contribution in $E(w_{f_{IP1} } \vert \tilde {k})$. Thus,
 the increment of $E(w_{f_{IP1} } \vert r_k \rightarrow s_j)$ is lager than that of $E(w_{f_{IP1} } \vert \tilde {k})$ when $r_i$ increases its
 bid from $v_i$ to $b_i$. $r_k$ will also be allocated in $s_j$ when $r_i$ bids $b_i>v_i$. In conclusion, $r_l$ cannot be allocated in $s_j$ when $r_i$ bids $b_i>v_i$.

Based on the analysis above, if $r_i $ wins the auction with a bid $v_i $, it always wins with the bid $b_i>v_i$.
\end{proof}


\begin{theorem}
MDCA can be executed in polynomial time.
\end{theorem}
\begin{proof}
We have shown that LP(2) can be solved in polynomial time. For each request with $x_i \ge 0$ in the optimal solution of LP(2), we solve LP(2) no more than $m$ times to check if request $r_i$ can be allocated in a channel. Then, the time complexity of MDCA is $O(nm)$ multiplied by the time complexity of solving LP(2). This finishes the proof.
\end{proof}

Since the revised Algorithm MDCA can be executed in polynomial time, we can find the critical value for each winner using a method such
as binary search. However, the time complexity of binary search is depends on the ratio of the max bid among requests to the bid size, which may exponential times of $n$. It is hard to find the critical values for winners in polynomial times. Thus, we further design another channel allocation mechanism that is truthful in expectation.

\subsubsection{Truthful in expectation}

In this section, we will employ a technique proposed by Lavi and Swamy \cite{nearoptimallinear} to design a channel allocation mechanism (CATE), which achieves the truthfulness in expectation.

The basic idea is depicted as follows. With the optimal solution of LP(2), ${\rm {\bf X}} = (x_i )_n $, we can get a set of feasible
solutions of IP(1), $\mathcal{L}$, by allocating some requests that $x_i\ge 0$ in channels. For each feasible solution
$l \in \mathcal{L}$, we define a probability $q(l)$ which will be discussed later, and choose $l$ as the final solution
 with probability $q(l)$. Let $x^{l}_{i}=1$ denote that request $r_i$ wins in solution $l$, and let $x^{l}_{i}=0$
 denote that $r_i$ loses. Then, if the equation $\sum\nolimits_{l\in \mathcal{L}}{x_{i,l}q(l)}=\frac{x_i}{\alpha}$ is
 established for any request $r_i$, we can get that the probability of request $r_i$ being assigned a channel is exactly
  $\frac{x_i}{\alpha}$. For each winner, the charge can be calculated as follows

\begin{equation}
\label{eq11}
p_i= \frac{1}{x_i}(\sum\nolimits_{j\ne i}{\phi_j x'_j}-\sum\nolimits_{j\ne i}{\phi_j x_j})
\tag{11}
\end{equation}

The vector $X'=(x'_j)_n$ is obtained by computing LP(2) with $b_i=0$. We show that this
allocation mechanism and payment scheme result in an auction, which is truthful in expectation.

\begin{theorem}
CATE is truthful in expectation.
\end{theorem}
\begin{proof}
Let $u_i(b_i)$ be the utility of request $r_i$ when bidding with $b_i$. Then, the expected utility of $r_i$ is

\setcounter{equation}{11}
\begin{equation}
\begin{array}{r}
\displaystyle E[u_i(b_i)]= \frac{x_i}{\alpha} [v_i - \frac{1}{x_i} (\sum\nolimits_{j \ne i}{\phi_j x'_j} - \sum\nolimits_{j \ne i}{\phi_j x_j})] \\
\displaystyle = \frac{1}{\alpha}[v_i x_i+ \sum\nolimits_{j \ne i}{\phi_j x_j}- \sum\nolimits_{j \ne i}{\phi_j x'_j}]
\end{array}
\end{equation}

Since $\sum\nolimits_{j\ne i}{\phi(j)x'_j}$ keeps unchanged when we increase or decrease
the bid of $r_i$, $E[u_i(b_i)]$ is maximized when $b_i=v_i$. That means the expected
utility of $r_i$ is maximized when $r_i$ bids truthfully. \end{proof}

The distribution of $P(l)$ can be solved by the following LP.

\begin{equation*}
\setlength{\abovedisplayskip}{1pt}
\setlength{\belowdisplayskip}{1pt}
\min \sum\nolimits_{l\in \mathcal{L}}{q(l)},
\tag{LP(3)}
\end{equation*}
subject to
\begin{equation*}
\begin{cases}
\sum\limits_{l\in \mathcal{L}}{x^{l}_{i}q(l)}=\frac{x_i}{\alpha},  \forall r_i\in \mathcal{R'} \\
     \sum\limits_{l\in \mathcal{L}}{q(l)}\ge 1  \\
     q(l)\ge 0  , \forall l\in \mathcal{L} \\
\end{cases}
\end{equation*}

The dual of LP(3) is:
\begin{equation*}
\setlength{\abovedisplayskip}{1pt}
\setlength{\belowdisplayskip}{1pt}
\max  z+\sum\nolimits_{r_i\in \mathcal{R'}}{\frac{x_i}{\alpha} w_i},
\tag{LP(4)}
\end{equation*}
subject to
\begin{equation*}
\begin{cases}
z+\sum\limits_{r_i\in \mathcal{R'}}{x^{l}_{i} w_i} \leq 1 ,  \forall l\in \mathcal{L} \\
     z \ge 0    \\
\end{cases}
\end{equation*}

Since LP(3) has an exponential number of variables, we discuss its dual (LP(4)). LP(4) has an exponential number of constraints, and we can view $w$ in LP(4) as a valuation. Suppose a $\alpha$-approximation algorithm $App$ proves an integrality gap of $\alpha$ with the optimal solution of LP(2). It has been shown in \cite{nearoptimallinear} that a separation oracle for LP(4) can be obtained by using Algorithm $App$ with valuation $w$, so the ellipsoid method can be used to solve LP(4) and hence LP(3). In CATE, we choose the allocation method DCA which is designed in the last section as $App$. Then, we can get that $\alpha = \frac {e}{e-1}$. Since the probability of any request $r_i$ being assigned a channel is exactly $\frac{x_i}{\alpha}$, we can conclude that the expected weight of the solution of CATE is larger than $1-1/e$ times of the weight of the optimal solution of IP(1).

\begin{theorem}
CATE can be executed in polynomial time.
\end{theorem}
\begin{proof}
We can use ellipsoid method and mechanism DCA on LP(4) with the optimal solution of LP(2) to compute a set of feasible solutions of IP(1). Since the ellipsoid method takes at most polynomial number of steps, and the mechanism DCA can be executed in polynomial time, thus, they can be used to return a set of solutions $\mathcal{L}$ in a polynomial size. Obviously, LP(3) can also be solved in polynomial time with $\mathcal{L}$ in a polynomial size. This finishes the proof.
\end{proof}

\section{Simulation Results}\label{sec:simulation}
The main purpose of our extensive simulations is to examine the performance of the proposed auctions. We first start by describing our simulation setup. Then, we study the setting variance impact on the performance of the proposed auction mechanisms.

\subsection{Simulation Setup}
In our simulation, we assume there is only one primary user who subleases the usage of $3$ channels in the spectrum market,
and the auction period $T$ is one hour. We use the disk model to simulate the license area of each candidate spectrum,
and the radius of license area is randomly generated from $40$ to $70$. All the buyers are randomly distributed within
a fixed area of $100 \times 100$ square units. Without loss of generality, we further assume that all the buyers' bid values
are uniformly, exponentially or Gaussian distributed in $[0,1]$, and the time duration $t_i$ for each request $r_i$ is randomly
generated from $10$ to $30$ minutes.

\begin{figure*}
  \centering
  \begin{minipage}[t]{0.33\textwidth}\label{fig:SocialUniformdistribution}
    \centering
    \includegraphics[width=6cm,height=4cm]{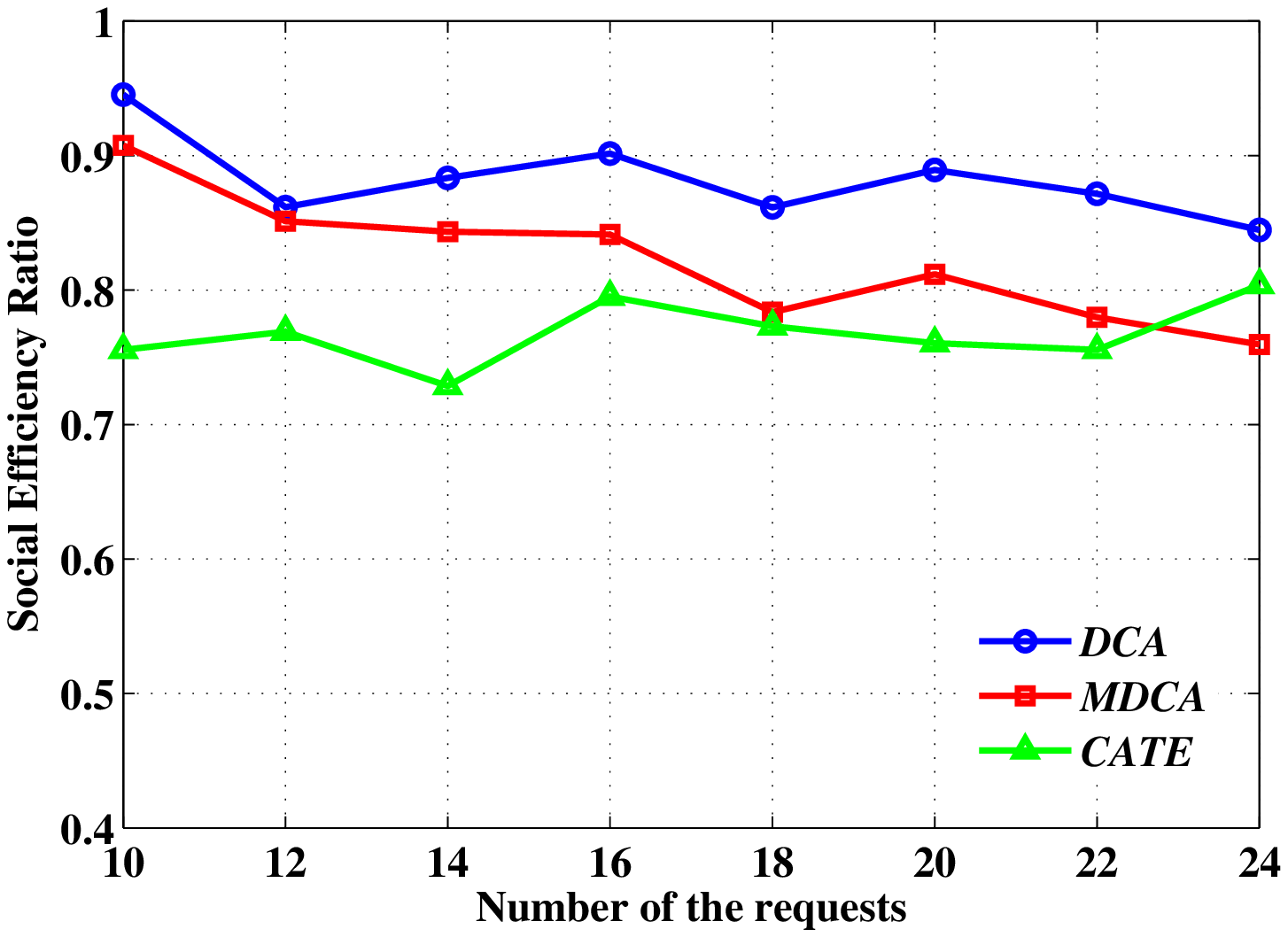}
    \small{\caption{Social Efficiency under Uniform Distribution}}
  \end{minipage}%
  \begin{minipage}[t]{0.33\textwidth}\label{fig:SocialExponentdistribution}
    \centering
    \includegraphics[width=6cm,height=4cm]{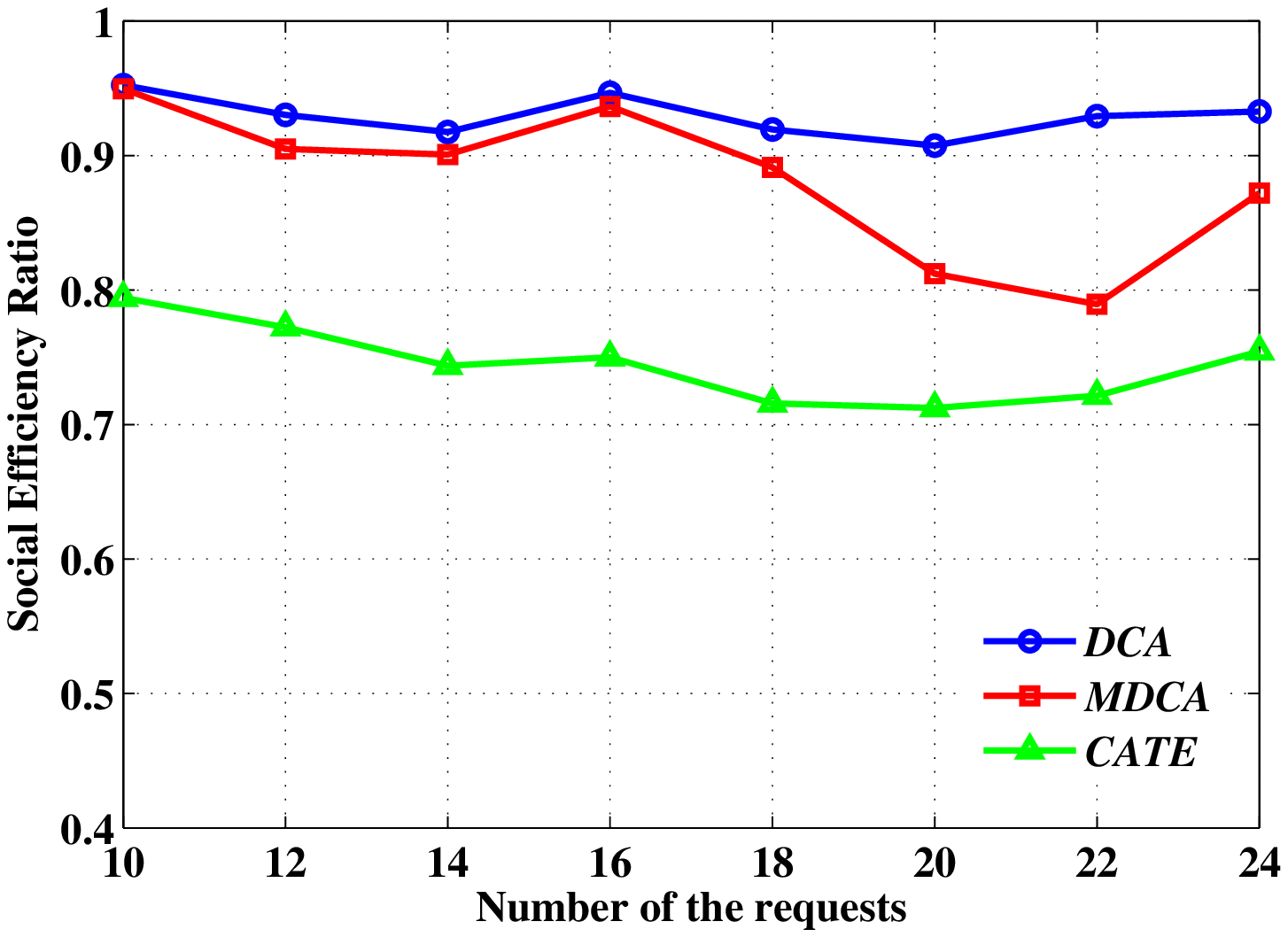}
    \small{\caption{Social Efficiency under Exponential Distribution}}
  \end{minipage}%
  \begin{minipage}[t]{0.33\textwidth}\label{fig:SocialGaussiandistribution}
    \centering
    \includegraphics[width=6cm,height=4cm]{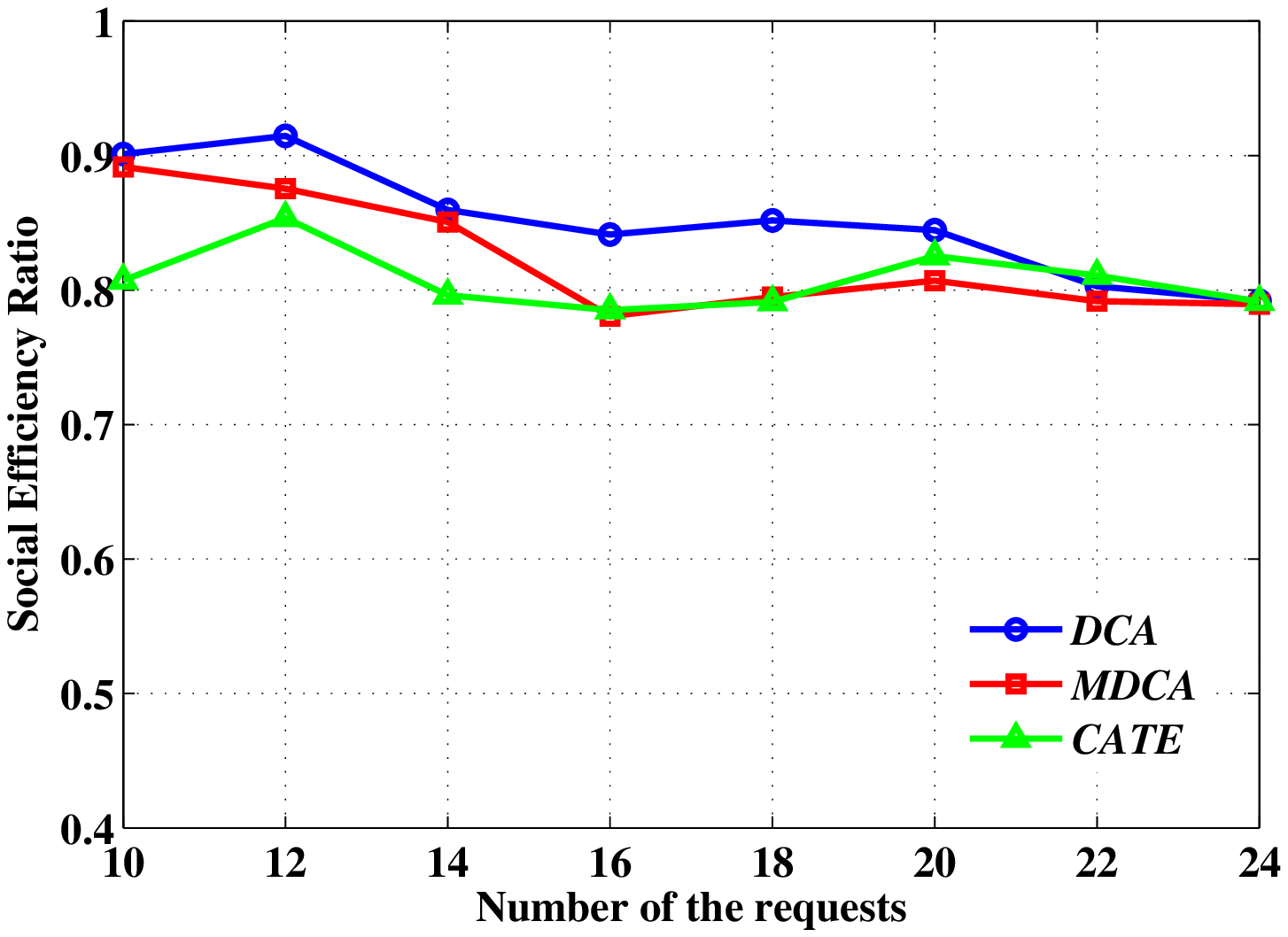}
    \small{\caption{Social Efficiency under Gaussian Distribution}}
  \end{minipage}%
\end{figure*}

\begin{figure*}
  \centering
  \begin{minipage}[t]{0.33\textwidth}\label{fig:RevRatioUni}
   \centering
    \includegraphics[width=6cm,height=3.5cm]{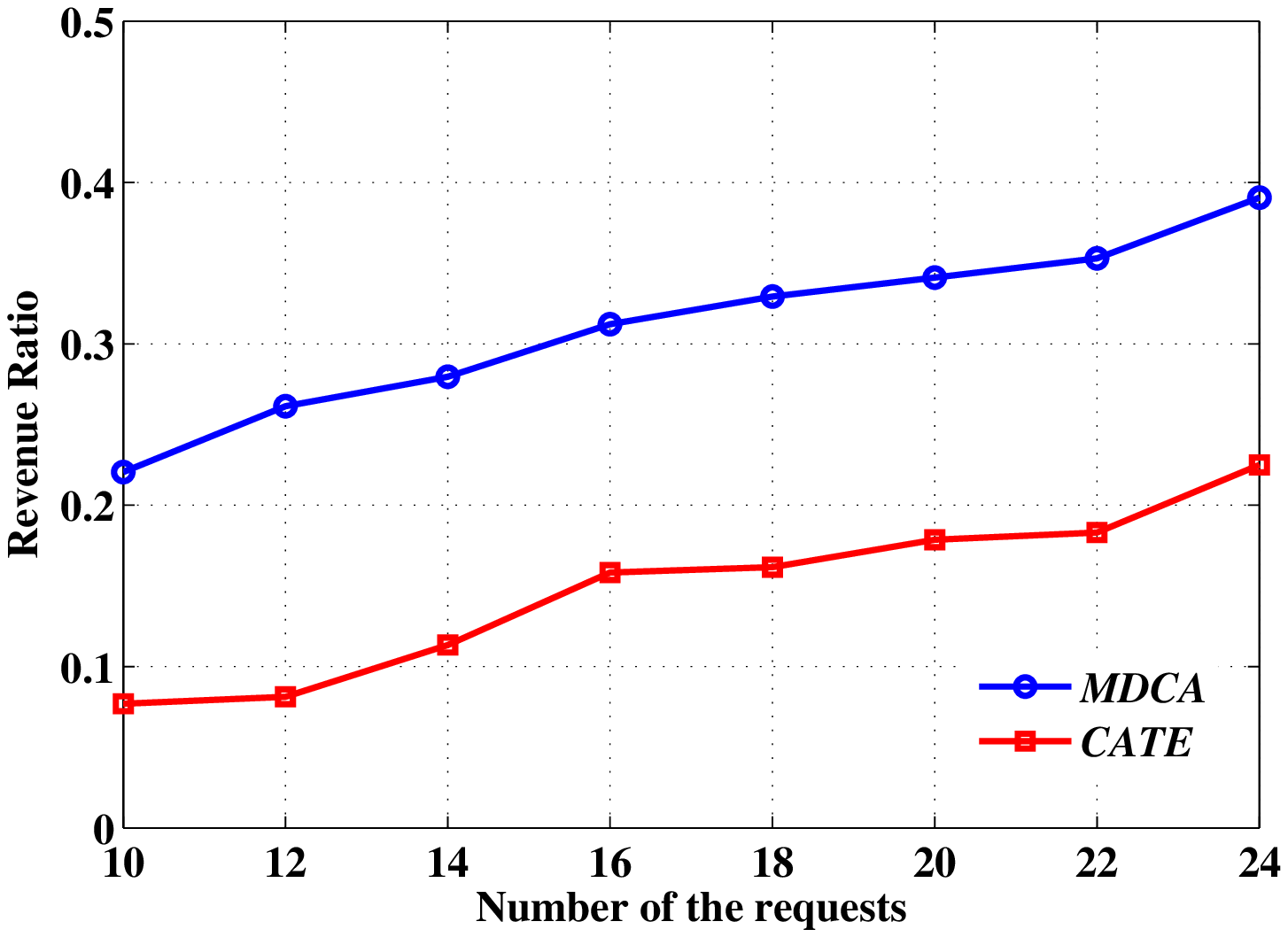}
    \small{\caption{Revenue Ratio under Uniform Distribution}}
  \end{minipage}%
   \begin{minipage}[t]{0.33\textwidth}\label{fig:RevRatioExp}
    \centering
    \includegraphics[width=6cm,height=3.5cm]{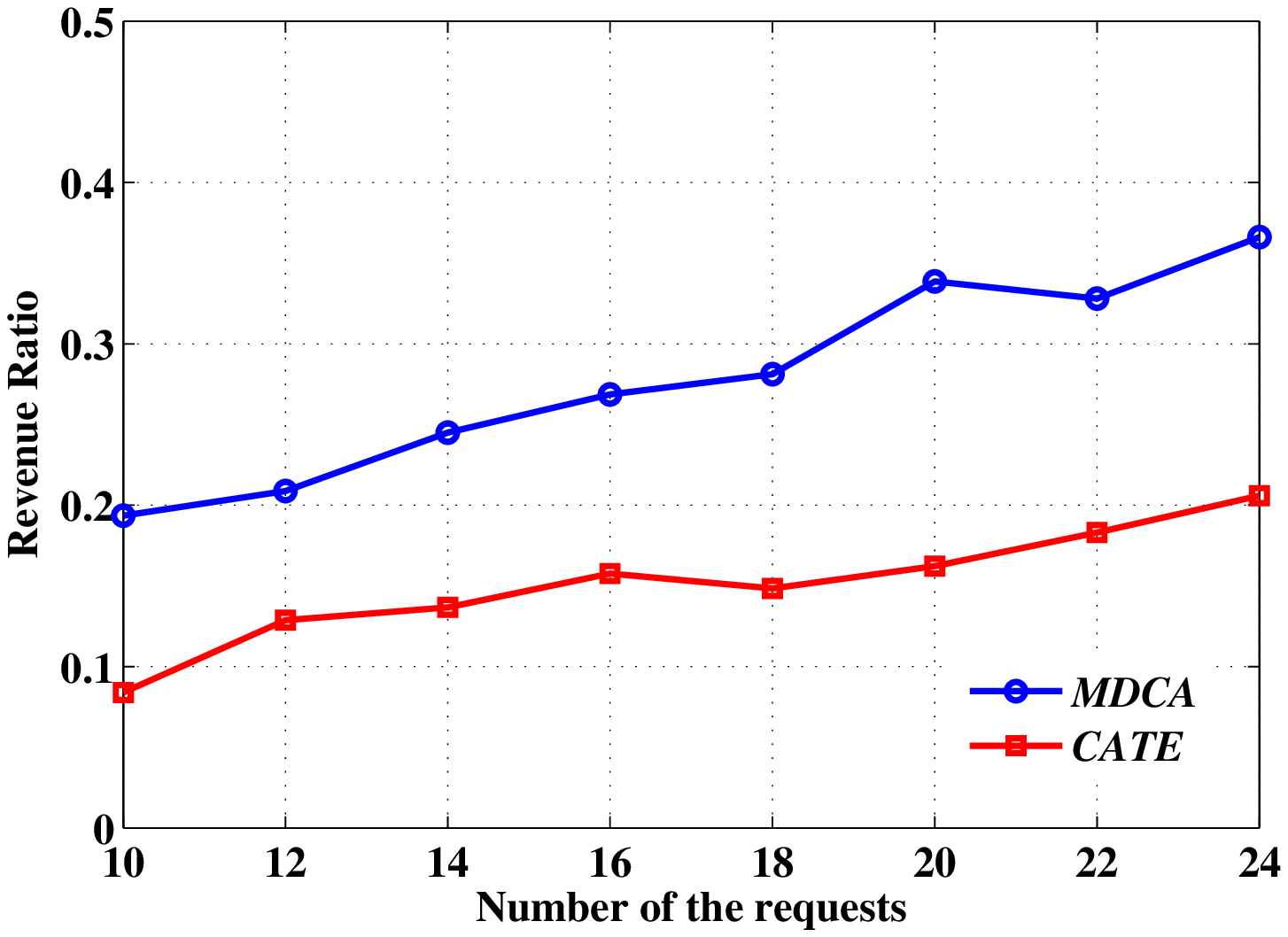}
    \small{\caption{Revenue Ratio under Exponential Distribution}}
  \end{minipage}%
  \begin{minipage}[t]{0.33\textwidth}\label{fig:RevRatioGau}
    \centering
    \includegraphics[width=6cm,height=3.5cm]{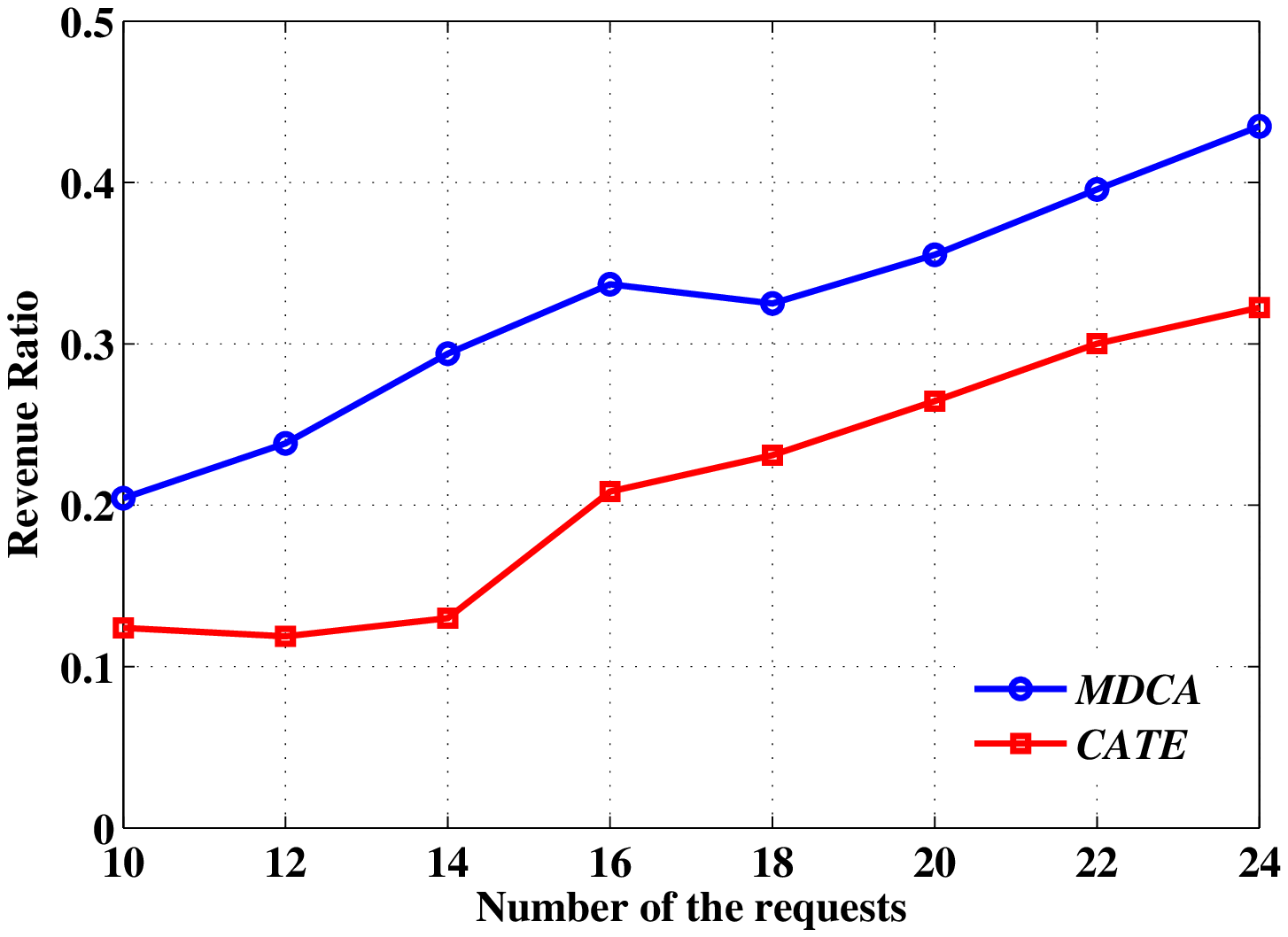}
    \small{\caption{Revenue Ratio under Gaussian Distribution}}
  \end{minipage}%
\end{figure*}

\subsection{Performance analysis}
We set the interference radius for each channel to be equal to $30$, and run our mechanisms under three types of bids
distributions (uniform, exponential and Gaussian) in Figures $2$-$7$.

In Figures $2$-$4$, we plot the social efficiency ratio of our channel allocation mechanisms (DCA, MDCA and CATE)
and the optimal allocation mechanism. Unsurprisingly,
the performances of DCA and MDCA are better than those of CATE. That's because DCA and MDCA always get a solution
whose value is larger than $1-1/e$ times of the optimal one, while the solutions of CATE may be arbitrarily bad. However,
our simulation results are much better than the theoretical bound we proved
in the previous sections. Even the solution of CATE is always larger than $70\%$ of the optimal solution.

From Figures $2$-$4$, we can also learn that the social efficiency ratio is declined slightly with the increasing number
of requests. The reason for that is most of requests can be allocated in channels without conflicting with others when
there are only few requests. Thus, our approximate auction mechanisms (DCA, MDCA and CATE) perform almost as well as the optimal one.
Since the competition among requests increases as the number of requests increases, the optimal auction mechanism outperforms DCA,
MDCA and CATE gradually. The social efficiency ratio keeps approximately stable when the number of requests is large enough,
that is, the supply is much less than the demand.

Observe that computing the optimal revenue is an NP-hard problem and the optimal social efficiency is an upper bound of revenue,
we define the revenue ratio to be the ratio of the total payments of winners and the optimal social efficiency. Since we cannot
prove that DCA is bid monotone, we cannot compute the critical payment for each winner either. Therefore, we only plot the revenue ratio of
MDCA and CATE in Figures $5$-$7$. As shown in these figures, the revenue ratio of primary user is increased with the number of requests
when the reservation price keeps unchanging. That is because the payment of each winner in our auction mechanisms is its critical value,
which is increased with the competition among requests.





\section{Literature Reviews}\label{sec:review}
Auction theory, regarded as a subfield of economics and game theory, serves as an efficient, fair way to distribute scarce resources amongst competing users. Recent years, auctions have been extensively studied in the scope of spectrum allocation. Many studies on spectrum auctions \cite{gandhi2007general}, \cite{huang2006auction}, \cite{ryan2006new} have been proposed to cope with the dynamic spectrum access problem in different perspectives on optimization goal, such as maximizing the total utility or minimizing the spectrum interference.

Truthfulness (or strategyproofness) is considered as one of the most critical factors to attract participation in the design of auction mechanism. Nevertheless, none of the earlier studies address this issue. Although large amount of studies are designed aiming at achieving economical robustness (\eg \cite{babaioff2001concurrent}, \cite{mcafee1992dominant}, \cite{vickrey1961counterspeculation}), these traditional auctions will lose the truthful property when they are directly applied to spectrum auctions due to some constraints, such as spatial and temporal reuse of spectrum. Meanwhile, some well known auction mechanisms (\emph{such as VCG}) will lose the truthfulness when applied to suboptimal algorithms.
In \cite{zhou2008ebay}, truthfulness is first designed for spectrum auction, where the spatial reuse is considered. \cite{al2011truthful} and \cite{jia2009revenue} focus on maximizing revenue for the auctioneer; \cite{gopinathanstrategyproof} studies the fairness and economic feasibility in spectrum auction model to achieve the global fairness and truthfulness; Zhou \etal \cite{zhou2009trust} first takes the McAfee double auction model into spectrum allocation to achieve the economic robustness.

Spectrum is a local resource, and it usually trades within its license region in a secondary markets. \emph{District} mechanism \cite{wangdistrict} first takes the spectrum locality into consideration and proposes an economically robust double auction method. \cite{fengtahes} proposes a truthful auction model for heterogeneous spectrum trading with the consideration of spectrum locality. As another line of spectrum reuse, \cite{deek2011preempt}, \cite{wang2010toda}, \cite{xu2011efficient}, \cite{xu2011truthful} study the spectrum allocation in an online model. The temporal reuse is adopted in these online model researches. However, combination of spectrum locality and temporal reuse is not fully mentioned in the previous studies. Dong \etal \cite{dong2012combinatorial} tackles spectrum auction by introducing combinatorial auction model which achieves time-frequency flexibility, however, the authors do not consider spatial reuse and spectrum locality property in their work.
Our work essentially generalizes all of the above challenges in the auction design. 


\section{Conclusion}\label{sec:conclusion}
In this paper, we have studied the case that spectrum can not only be reused in spatial domain, but also in
temporal domain. We have designed a general truthful spectrum auction framework which can maximize
the social efficiency or revenue. As allocating channels optimally is NP-hard in our model, we have also proposed
a set of near-optimal channel allocation mechanisms with $(1-1/e)$ performance guarantee.

Several interesting questions are left for future research. The first one is  to design a spectrum auction
mechanism that can guarantee a good approximation and an efficient practical running time. The second one
is to relax the time request model from the fixed interval model we studied in this paper to a more general one.
The third challenging question is  to design truthful mechanisms with good performance guarantee when we
have to make online decisions.

\section*{Acknowledgement}

The research of authors is partially supported by the National Grand
Fundamental Research 973 Program of China (No.2011CB302905,
No.2011CB302705), National Natural Science Foundation of China (NSFC)
under Grant No. 61202028, No. 61170216, No. 61228202, and NSF
CNS-0832120, NSF CNS-1035894, NSF ECCS-1247944, Specialized Research Fund for the Doctoral Program of Higher Education (SRFDP) under Grant No. 20123201120010.

{\small
\bibliography{auction}
}
\end{document}